\documentclass[onecolumn,draftcls]{IEEEtran}

\def\BibTeX{{\rm B\kern-.05em{\sc i\kern-.025em b}\kern-.08em
    T\kern-.1667em\lower.7ex\hbox{E}\kern-.125emX}}

\def\abstract{ \noindent \small\bf Abstract: \hspace{3mm} \small }

\def\@begintheorem#1#2{\tmpitemindent\itemindent\topsep 0pt\rm\trivlist
    \item[\hskip \labelsep{\indent\it #1\ #2:}]\itemindent\tmpitemindent}
\def\@opargbegintheorem#1#2#3{\tmpitemindent\itemindent\topsep 0pt\rm \trivlist
    \item[\hskip\labelsep{\indent\it #1\ #2\ \rm(#3)}]\itemindent\tmpitemindent}
\def\@endtheorem{\endtrivlist\unskip}

\usepackage{times}
\usepackage{graphics,graphicx,fancyhdr,amsfonts,amsmath,color,epic}
\usepackage{graphicx,graphics}
\DeclareGraphicsExtensions{.jpg,.pdf,.eps,.png}

\newtheorem {lemma}{Lemma}

\begin{document}

\title{A Design of Paraunitary Polyphase Matrices of Rational Filter Banks Based on $(P,Q)$ Shift-Invariant Systems}

\author{Sudarshan Shinde \\
        Computational Research Laboratories \\
        Pune-04,INDIA. \\
        Email:sudarshan\_shinde@iitbombay.org}
\date{}

\maketitle



\pagestyle{plain}

\begin{abstract}
In this paper we present a method to design paraunitary polyphase matrices of critically sampled rational filter banks. The method is based on $(P,Q)$ shift-invariant systems, and so any kind of rational splitting of the frequency spectrum can be achieved using this method. Ideal $(P,Q)$ shift-invariant system with smallest $P$ and $Q$ that map of a band of input spectrum to the output spectrum are obtained. A new set of filters is obtained that characterize a $(P,Q)$ shift-invariant system. Ideal frequency spectrum of these filters are obtained using ideal $(P,Q)$ shift-invariant systems. Actual paraunitary polyphase matrices are then obtained by minimizing the stopband energies of these filters against the parameters of the paraunitary polyphase matrices.   
\end{abstract}


\section{Introduction}
\label{sec_intro}

\subsection{Motivation} 

Rational Filter Banks (RFB) are now well known tools for nonuniform subband decomposition of the signal. Given $N$ rational numbers $\{\frac{p_{0}}{q_{0}},\cdots, \frac{p_{N-1}}{q_{N-1}}\}$, where $\sum_{i=0}^{N-1}{\frac{p_{i}}{q_{i}}} = 1$, the task of an $N$ channel RFB is to split the frequency spectrum $[0,\pi]$ of the signal into $N$ nonuniform bands, where $n$-th band is given by $[\sum_{i=0}^{n}{\frac{p_{i}}{q_{i}}}\pi,\sum_{i=0}^{n+1}{\frac{p_{i}}{q_{i}}}\pi]$. 

An RFB uses a $P/Q$ sample rate changer, shown in Fig.(\ref{fig_rational_branch}) in each channel. 

\begin{figure*}[t]
\centerline{\includegraphics[scale=0.4]{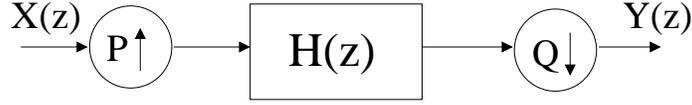}}
\caption{A typical branch of a rational filter bank}
    \label{fig_rational_branch}
\end{figure*}


Such RFB are treated in detail in \cite{art_vetterli_rationalFB} and it is shown that all kind of rational band splitting can not be achieved by the RFB based on the above $P/Q$ sample rate changer, even if the ideal real coefficients filters are used. 

If $P$ and $Q$ are coprime, the structure shown in Fig.(\ref{fig_rational_branch}) can be converted into a structure shown in Fig.(\ref{fig_gen_poly}). In this figure ${\bf H}_{p}(z)$ is a $P \times Q$ polyphase matrix that consists of polyphase components of $H(z)$. It is shown in \cite{art_shenoy_multirate} that this structure is more general than the previous structure, and if $P$ and $Q$ are allowed to have common factors, then any kind of rational splitting of the input frequency spectrum can be achieved.

\begin{figure*}[t]
\centerline{\includegraphics[scale=0.4]{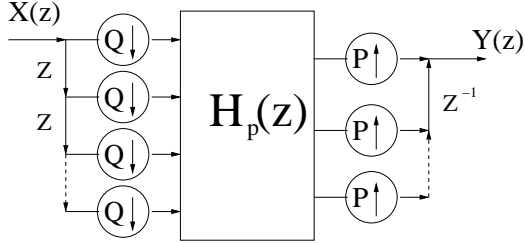}}
\caption{Polyphase structure for RFB implementation}
    \label{fig_gen_poly}
\end{figure*}

Such structures are called $(P,Q)$ shift-invariant systems, since shifting the input by $Q$ samples results in shifting the output by $P$ samples. Many properties of $(P,Q)$ shift-invariant systems are studied in \cite{art_chen_general}. It is shown that such structures can be represented by a $P \times Q$ LTI polyphase matrix ${\bf H}_{p}(z)$. In other words, let input $x(z)$ and output $y(z)$ be written in polyphase form as $x(z) = \sum_{i=0}^{Q-1}{z^{-i}x_{i}(z^{Q})}$ and $y(z) = \sum_{i=0}^{P-1}{z^{-i}y_{i}(z^{P})}$ respectively, and let
 
\begin{eqnarray}
{\bf x}_{p}(z) &=& [x_{0}(z),x_{1}(z),\cdots,x_{Q-1}(z)]^{T} \nonumber \\ 
{\bf y}_{p}(z) &=& [y_{0}(z),y_{1}(z),\cdots,y_{P-1}(z)]^{T} 
    \label{eq_polyphase_xy}
\end{eqnarray} 
then
\begin{equation}
{\bf y}_{p}(z) = {\bf H}_{p}(z) {\bf x}_{p}(z).
    \label{eq_polyphase}
\end{equation}  

Let $S_{x}(\omega_{1},\omega_{2})$ represent the frequency spectrum of the signal $x(n)$ between frequencies $\omega_{1} \le \omega \le \omega_{2}$ . The $(P,Q)$ shift-invariant systems described above can be used to obtain the ideal mapping of an input subband $S_{x}(\frac{p_{1}}{q_{1}}\pi,\frac{p_{2}}{q_{2}}\pi)$ on the output band $S_{y}(0,\pi)$ or $S_{y}(-\pi,0)$ \cite{art_shenoy_multirate}. Let 
\begin{eqnarray}
{\bf x}_{m}(z) &=& [x(z),x(zW_{Q}),\cdots,x(zW_{Q}^{Q-1})]^{T} \nonumber \\
{\bf y}_{m}(z) &=& [y(z),y(zW_{P}),\cdots,y(zW_{P}^{P-1})]^{T} 
    \label{eq_modulation_xy}
\end{eqnarray} 
then it can be shown that
\begin{eqnarray}
{\bf x}_{p}(z) 
    &=& \frac{1}{Q}{\bf \Gamma}(z^{1/Q}){\bf W}_{Q}{\bf x}_{m}(z^{1/Q})     
                                           \nonumber \\
{\bf y}_{p}(z) 
    &=& \frac{1}{P}{\bf \Gamma}(z^{1/P}){\bf W}_{P}{\bf y}_{m}(z^{1/P})
    \label{eq_poly_to_mod_2}
\end{eqnarray}   
where $[{\bf \Gamma}(z)]_{ij} = z^{i}\delta_{ij}$ and $[{\bf W}_{N}]_{ij} = W_{N}^{ij}$. Substituting from (\ref{eq_poly_to_mod_2}) into (\ref{eq_polyphase}), we get 

\begin{equation}
{\bf y}_{m}(z^{1/P}) = {\bf H}_{m}(z^{1/K}){\bf x}_{m}(z^{1/Q})
    \label{eq_poly_to_mod_3}
\end{equation} 
where $K = lcm(P,Q)$, and ${\bf H}_{m}(z^{1/K})$ will be given by 

\begin{equation}
{\bf H}_{m}(z^{1/K}) 
              = \frac{1}{Q} {\bf W}_{P}^{\dagger}{\bf \Gamma}_{P}(z^{-1/P})
                             {\bf H}_{p}(z)
                             {\bf \Gamma}_{mq}(z^{1/Q}){\bf W}_{Q}   
    \label{eq_h_mod_2} 
\end{equation} 

${\bf H}_{m}(z)$ is called alias component matrix, or modulation matrix. It can be seen that (\ref{eq_poly_to_mod_3}) gives a relationship between input subbands and output subbands. By proper choice of the dimensions $(P,Q)$ of the modulation matrix ${\bf H}_{m}(z)$ and by proper selection of the elements of ${\bf H}_{m}(z)$ from $\{0,1\}$ any required mapping can be obtained. It is shown in \cite{art_shenoy_multirate} that mapping $S_{x}(\frac{p_{1}}{q_{1}},\frac{p_{2}}{q_{2}}) \rightarrow S_{y}(0,\pi)$ can be obtained by using a $(P,Q)$ shift-invariant system, or a $P \times Q$ modulation matrix, where $Q = 2lcm(q_{1},q_{2})$ and $P = Q\left(\frac{p_{2}}{q_{2}} - \frac{p_{1}}{q_{1}} \right)$. A modulation matrix ${\bf H}_{m}(z)$ consisting of elements from $\{0,1\}$ will be referred as Ideal Modulation Matrix (IMM).

In this paper we consider the problem of obtaining the paraunitary polyphase matrices of the filter banks that can be realized by $(P,Q)$ shift-invariant systems. In the next subsection we define the problem more precisely and highlight the contributions of the paper.

\subsection{Problem definition and contributions of the paper}
Let us say that we are given any $N$ rational numbers $\{\frac{p_{1}}{q_{1}},\cdots,\frac{p_{N-1}}{q_{N-1}}\}$. It is required to split the input frequency band $[0,\pi]$ into $N$ bands, such that $n$-th channel gives the band $[\sum_{i=0}^{n}{\frac{p_{i}}{q_{i}}}\pi,\sum_{i=0}^{n+1}{\frac{p_{i}}{q_{i}}}\pi]$. Our objective is to design a paraunitary polyphase matrix of minimum dimensions that achieves this splitting. In order to achieve this objective, we address the following issues: 
\begin{enumerate}
     \item If an input subband 
           $S_{x}(\frac{p_{1}\pi}{q_{1}},\frac{p_{2}\pi}{q_{2}})$ is to be mapped to
           output band $S_{y}(0,\pi)$, or $S_{y}(-\pi,0)$, what would be the 
           smallest dimensions of the ideal modulation matrix ?
           We show in section \ref{sec_ideal_modulation_mat} that though the
           above mapping can be obtained by a $P \times Q$ ideal modulation
           matrix, where $Q = 2lcm(q_{1},q_{2})$ and 
           $P = Q\left(\frac{p_{2}}{q_{2}} - \frac{p_{1}}{q_{1}} \right)$,
           as given in \cite{art_shenoy_multirate}, these are not the 
           smallest dimensions in all the cases. In the same section we obtain
           smallest dimension ideal modulation modulation matrices for a
           given  mapping.    
     \item We prove that a $(P,Q)$ shift-invariant system can be characterized
           by $gcd(P,Q)$ filters. This result as such is not new. Many 
           structures that realize $(P,Q)$ shift-invariant systems using 
           $gcd(P,Q)$ filters are given in \cite{art_mehr_representations}.
           However we could not find a method to use these structures
           to design an RFB based on these structures. 
           In section \ref{sec_modulation_mat} we obtain a set 
           of $gcd(P,Q)$ filters
           that is different from the filters given in 
           \cite{art_mehr_representations}. We show that using the ideal
           modulation matrices, we can obtain ideal frequency spectrums of 
           these filters. These ideal frequency spectrums are then used to 
           form an objective function. By minimizing this objective function
           we obtain the required paraunitary polyphase matrix. 
\end{enumerate}  

\subsection{Notations} 
The notations that we use in this paper are as follows. The frequency spectrum of a signal $x(n)$ between $\omega_{1}$ and $\omega_{2}$ will be denoted by $S_{x}(\omega_{1},\omega_{2})$. All vectors will be represented by boldface letters and matrices will be represented by upper boldface letters. For two integers $k$ and $m$, $k$ {\it modulo} $m$ is denoted by $(k)_{m}$. For a matrix ${\bf H}$, the $lm$-th element is denoted by $[{\bf H}]_{lm}$. A set of integers from $N_{1}$ to $N_{2}$ will be denoted by $\{N_{1}.. N_{2}\}$. For a polynomial matrix ${\bf D}(z)$, $\tilde{\bf D}(z) = {\bf D}^{T}_{*}(z^{-1})$, where the $*$ subscript indicates taking complex conjugate of polynomial coefficients. $\chi(\omega_{1},\omega_{2})$ denote a frequency spectrum that is $1$ for $[\omega_{1},\omega_{2}]$ and zero elsewhere.

\section{Structure of the Ideal Modulation Matrices}
\label{sec_ideal_modulation_mat}

We recall that the modulation domain equation is given by (\ref{eq_poly_to_mod_3}). Putting $z = exp(j\omega)$ in this equation, we get
\begin{equation}
{\bf y}_{m}(e^{j\omega/P}) 
       = {\bf H}_{m}(e^{j\omega/K}){\bf x}_{m}(e^{j\omega/Q})
    \label{eq_poly_to_mod_3_omega}
\end{equation} 
It can be noted that for $\omega \in [-\pi,\pi]$, the elements of ${\bf y}_{m}(e^{j\omega/mp})$ cover disjoint segments of $S_{y}(-\pi,\pi)$ and the union of these segments is equal to $S_{y}(-\pi,\pi)$. Consider a subset $(\omega_{1},\omega_{2})$ of $[-\pi,\pi]$. For this subset we define an ideal modulation matrix ${\bf H}_{I}(\omega_{1},\omega_{2})$ to be a matrix that maps some segments of $S_{x}$ on the respective segments of $S_{y}$ according to given desired mapping when $\omega \in (\omega_{1},\omega_{2})$. We say that a mapping is possible if we can obtain disjoint subsets of $[-\pi,\pi]$ such that their union is $[-\pi,\pi]$ and for each subset, we can obtain an ideal modulation matrix.

{\it Example :} Consider the following mapping

\begin{eqnarray}
S_{x}(\frac{2\pi}{5},\pi) 
 & \rightarrow & S_{y}(0,\pi) \nonumber \\
S_{x}(-\pi,-\frac{2\pi}{5}) 
 & \rightarrow & S_{y}(-\pi,0) 
    \label{eq_map_ex}
\end{eqnarray}
For this mapping we have ${\bf H}_{I}(-\pi,0)$ and ${\bf H}_{I}(0,\pi)$ are given by
\begin{eqnarray}
{\bf H}_{I}(-\pi,0) & = & \left[
                            \begin{array}{ccccc}
                            0 & 1 & 0 & 0 & 0 \\
                            0 & 0 & 1 & 0 & 0 \\
                            0 & 0 & 0 & 1 & 0 
                            \end{array}
                          \right] \nonumber \\ 
{\bf H}_{I}(0,\pi)  & = & \left[
                            \begin{array}{ccccc}
                            0 & 0 & 0 & 0 & 1 \\
                            0 & 0 & 1 & 0 & 0 \\
                            0 & 0 & 0 & 1 & 0 
                            \end{array}
                          \right]  
    \label{eq_ideal_mod_mat_ex}
\end{eqnarray}  
It can be seen that ${\bf y}_{m}(e^{j\omega/3}) = {\bf H}_{I}(-\pi,0) {\bf x}_{m}(e^{j\omega/5})$ where $\omega \in (-\pi,0)$ and ${\bf y}_{m}(e^{j\omega/3}) = {\bf H}_{I}(0,\pi) {\bf x}_{m}(e^{j\omega/5})$ where $\omega \in (0,\pi)$ together give the required mapping. 

We now consider the following problem  - given two rational numbers $\frac{p_{1}}{q_{1}}$ and $\frac{p_{2}}{q_{2}}$, where $gcd(p_{1},q_{1}) =1$ and $gcd(p_{2},q_{2}) =1$, find out the the smallest modulation matrices that achieve one of the following mappings:
\begin{eqnarray}
S_{x}(\frac{p_{1}\pi}{q_{1}},\frac{p_{2}\pi}{q_{2}}) 
 & \rightarrow & S_{y}(0,\pi) \nonumber \\
S_{x}(-\frac{p_{2}\pi}{q_{2}},-\frac{p_{1}\pi}{q_{1}}) 
 & \rightarrow & S_{y}(-\pi,0) 
    \label{eq_map_1}
\end{eqnarray} 
or
\begin{eqnarray}
S_{x}(\frac{p_{1}\pi}{q_{1}},\frac{p_{2}\pi}{q_{2}}) 
 & \rightarrow & S_{y}(-\pi,0) \nonumber \\
S_{x}(-\frac{p_{2}\pi}{q_{2}},-\frac{p_{1}\pi}{q_{1}}) 
 & \rightarrow & S_{y}(0,\pi) 
    \label{eq_map_2}
\end{eqnarray}  
 
Let $Q = lcm(q_{1},q_{2})$, $P_{1} = \frac{p_{1}Q}{q_{i}}$ and $P_{2} = \frac{p_{2}Q}{q_{2}}$. Let $P = P_{2} - P_{1}$, $m = gcd(P,Q)$, $p = P/m$, $q = Q/m$. The following lemma gives the smallest ideal modulation matrices.

\begin{lemma}
\label{thm_imm}
The ideal modulation matrices will have dimensions $P \times Q$ if and only if either of the $P_{1}$ or $P_{2}$ is even. If both $P_{1}$ and $P_{2}$ are odd, then the ideal modulation matrices will be of dimensions $2P \times 2Q$. Moreover the ideal modulation matrices are given as follows-
\begin{enumerate}
     \item $P_{1}$ is even:
           In this case we will have two modulation matrices 
           ${\bf H}_{I}(-\pi,0)$ and ${\bf H}_{I}(0,\pi)$. These matrices are
           given as follows -
           for $l \in [0,P-1]$, $[{\bf H}_{I}(-\pi,0)]_{lk} = 1$ if $k$ is 
           given by
           \begin{equation}
               k = \left\{ \begin{array}{lc}
                              \frac{P_{1}}{2} + l & 2l < P \\ 
                              (Q - (\frac{P_{1}}{2} + P - l))   & 2l \ge P 
                           \end{array}
                   \right.
               \label{eq_kl_relation_p1}
           \end{equation}  
           else $[{\bf H}_{I}(-\pi,0)]_{lk} = 0$. ${\bf H}_{I}(0,\pi)$  
           is obtained from $[{\bf H}_{I}(-\pi,0)]_{lk}$ as follows

	   \begin{equation}
              [{\bf H}_{I}(0,\pi)]_{lr} 
                      = [{\bf H}_{I}(-\pi,0)]_{(P-l)_{P}(Q-r)_{Q}} 
              \label{eq_lideal_rideal}
           \end{equation} 
 
     \item $P_{2}$ is even:
           In this case also we will have two modulation matrices 
           ${\bf H}_{I}(-\pi,0)$ and ${\bf H}_{I}(0,\pi)$. These matrices 
           are given as follows - 

           for $l \in [0,P-1]$, $[{\bf H}_{I}(-\pi,0)]_{lk} = 1$ if $k$ is 
           given by
           \begin{equation}
               k = \left\{ \begin{array}{lc}
                              (Q - (\frac{P_{2}}{2} -l)) & 2l < P \\ 
                              \frac{P_{1} - P}{2} + l     & 2l \ge P 
                           \end{array}
                   \right.
               \label{eq_kl_relation_p2}
           \end{equation}  
           else $[{\bf H}_{I}(-\pi,0)]_{lk} = 0$. ${\bf H}_{I}(-\pi,0)$ 
           is obtained by using (\ref{eq_lideal_rideal}).

     \item Both $P_{1}$ and $P_{2}$ are odd:
           Since the dimension of the ideal modulation matrices becomes 
           $2P \times 2Q$, this case can be considered similar to the 
           above two cases, with
           $P_{1},P_{2},P,Q$ replaced by $2P_{1},2P_{2},2P,2Q$ in 
           (\ref{eq_kl_relation_p1}) and (\ref{eq_kl_relation_p2}).

\end{enumerate} 
    \label{thm_mod_mat_dim}
\end{lemma}
\begin{proof}
We first prove that the number of columns in the modulation matrix must be a multiple of $Q$. Let the dimension of the modulation matrix be $R \times S$ and $L = lcm(R,S)$. The modulation equation will then be ${\bf y}_{m}(z^{1/R}) = {\bf H}_{m}(z^{1/L}){\bf x}_{m}(z^{1/S})$. We first prove that $S$ needs to be a multiple of $lcm(q_{1},q_{2})$, i.e. $Q$. The $s$-th term in ${\bf x}_{m}(z^{1/S})$ and $r$-th term in ${\bf y}_{m}(z^{1/R})$ will be $x(z^{1/S} W_{S}^{s})$ and $y(z^{1/R} W_{R}^{r})$ respectively. The term $x(z^{1/S} W_{S}^{s})$ covers the input spectrum $S_{x}(\frac{\theta - 2 \pi s}{S},\frac{\theta - 2 \pi (s-1)}{S})$ for $z \in (e^{j\theta},e^{j(\theta + 2\pi)})$. If $[{\bf H}(z^{1/L})]_{rs} = 1$, then this input spectrum maps to the output spectrum $S_{y}(\frac{\theta - 2 \pi r}{R},\frac{\theta - 2 \pi (r-1)}{R})$. Consider now the mapping given in (\ref{eq_map_1}). In order to map the input frequency $\frac{p_{1}\pi}{q_{1}}$ to the output frequency $0$, we require $\theta$,$r$ and $s$ such that
 
\begin{equation}
\frac{\theta - 2 \pi s}{S} = \frac{p_{1}\pi}{q_{1}} 
    \label{eq_cond_x_rh}
\end{equation}  
and
\begin{equation}
\frac{\theta - 2 \pi r}{R} = 0
    \label{eq_cond_y_lh}
\end{equation}  

Substituting $\theta$ from (\ref{eq_cond_x_rh}) into (\ref{eq_cond_y_lh}), we obtain $\frac{Sp_{1}}{q_{1}} + 2(s-r) = 0$. Since $gcd(p_{1},q_{1}) = 1$, this condition will be fulfilled only if $S$ is a multiple of $q_{1}$. By a similar argument it can be shown that $S$ needs to be a multiple of $q_{2}$ in order to map the input frequency $\frac{p_{2}\pi}{q_{2}}$ to the output frequency $\pi$. Thus $S$ needs to be a multiple of both $q_{1}$ and $q_{2}$, or in other words, $S$ needs to be a multiple of $lcm(q_{1},q_{2})$, i.e. $Q$.

If $S = Q$, then the condition $\frac{Sp_{1}}{q_{1}} + 2(s-r) = 0$ reduces to $P_{1} + 2(s-r) = 0$. Thus if $S=Q$, and the mapping is given by (\ref{eq_map_1}) then we require $P_{1}$ to be even. By a similar argument, it can be shown that if $S=Q$, and the mapping is given by (\ref{eq_map_2}) then we require $P_{2}$ to be even. If both $P_{1}$ and $P_{2}$ are odd, then it is clear from this condition that $S$ is required to be equal to $2Q$.  

We now obtain ${\bf H}_{I}(-\pi,0)$ when $P_{1}$ is even. Consider the frequency $\frac{-2 l\pi}{P}$ in $S_{y}$. For $2 l \le P$, this frequency will be in $S_{y}(-\pi,0)$. The frequency in $S_{x}$ that should map to this frequency is 
\begin{equation}
\alpha = \frac{-P_{1}\pi}{Q} - \frac{2l\pi}{Q}
       = \frac{-2\pi}{Q}(\frac{P_{1}}{2} + l)   
    \label{eq_sx_sy_4}
\end{equation} 
Let $k$ be defined as
\begin{equation}
k = \frac{P_{1}}{2} + l
    \label{eq_xelem_yelem_3}
\end{equation}  
Then if $[{\bf H}_{I}(-\pi,0)]_{lk} = 1$, it would map $S_{x}(-\frac{(2k+1)\pi}{Q},-\frac{2k\pi}{Q})$ to $S_{y}(-\frac{(2l+1)\pi}{P},-\frac{2l\pi}{P})$.  

For $2l > P$, the frequency $\frac{- 2l\pi}{P}$ in $S_{y}$ will be equal to $2\pi - \frac{2l\pi}{P} = \frac{\pi(2P - 2l)}{P}$. The corresponding $S_{x}$ frequency is given by
\begin{equation}
\alpha = \frac{P_{1}\pi}{Q} + \frac{(2P - 2l)\pi}{Q} 
       = \frac{-2\pi}{Q}[Q - (\frac{P_{1}}{2} + P - l)]  
    \label{eq_sx_sy_5}
\end{equation} 
which gives $k$ as
\begin{equation}
k = [Q - (\frac{P_{1}}{2} + P - l)]  
    \label{eq_xelem_yelem_4}
\end{equation} 

It can be seen that ${\bf H}_{I}(-\pi,0)$ and ${\bf H}_{I}(0,\pi)$ together achieve the mapping given in (\ref{eq_map_1}).

The structure of the modulation matrices for rest of the two cases can be obtained in the similar manner.
\end{proof}   
{\it Remark:} Two results close to this lemma appear in \cite{art_vetterli_rationalFB}, proposition 3.2, and in \cite{art_shenoy_multirate}, theorem 1. However the result of this lemma is more general than the above two results and also unifies these results. Proposition 3.2 of \cite{art_vetterli_rationalFB} considers only the case when $P_{1}$ is even and $P > 1$. Theorem 1 of \cite{art_shenoy_multirate} always obtains ideal modulation matrices of dimensions $2P \times 2Q$, that are not the smallest ideal modulation matrices if $P_{1}$ or $P_{2}$ is even.    

{\it Example 1:} Consider again the mapping given in (\ref{eq_map_ex}). It can be seen easily that since in this case $P_{1}$ is even, the ideal modulation matrices given in (\ref{eq_ideal_mod_mat_ex}) can be obtained by (\ref{eq_kl_relation_p1}). 

{\it Example 2:} Consider the following mapping
\begin{eqnarray}
S_{x}(\frac{\pi}{3},\pi) 
 & \rightarrow & S_{y}(0,\pi) \nonumber \\
S_{x}(-\pi,-\frac{\pi}{3}) 
 & \rightarrow & S_{y}(-\pi,0) 
    \label{eq_map_ex2}
\end{eqnarray} 
Since both $P_{1}$ and $P_{2}$ are odd, the ideal modulation matrix dimensions will be $4 \times 6$, and the ideal modulation matrices can be given either from (\ref{eq_kl_relation_p1}) or from (\ref{eq_kl_relation_p2}). Considering (\ref{eq_kl_relation_p2}) for ideal modulation matrices, we obtain
\begin{eqnarray}
{\bf H}_{I}(-\pi,0)  &=& \left[
                            \begin{array}{cccccc}
                            0 & 0 & 0 & 1 & 0 & 0 \\
                            0 & 0 & 0 & 0 & 1 & 0 \\
                            0 & 0 & 0 & 0 & 0 & 1 \\
                            0 & 0 & 1 & 0 & 0 & 0 
                            \end{array}
                         \right]   \nonumber \\
{\bf H}_{I}(0,\pi)   &=& \left[
                            \begin{array}{cccccc}
                            0 & 0 & 0 & 1 & 0 & 0 \\
                            0 & 0 & 0 & 0 & 1 & 0 \\
                            0 & 1 & 0 & 0 & 0 & 0 \\
                            0 & 0 & 1 & 0 & 0 & 0 
                            \end{array}
                         \right]   
    \label{eq_ideal_mod_mat_ex2}
\end{eqnarray} 

If $P$ is even (i.e. both $P_{1}$ and $P_{2}$ are even, or both $P_{1}$ and $P_{2}$ are odd) then both the kinds of mappings given in (\ref{eq_map_1}) and (\ref{eq_map_2}) are possible. We now obtain the relationship of ideal modulation matrices for one mapping with those of the other mapping. Let the dimensions of ideal modulation matrices be $R \times S$. Then the modulation equation will be 
\begin{equation}
{\bf y}_{m}(e^{j\omega/R}) 
       = {\bf H}_{m}(e^{j\omega/K}){\bf x}_{m}(e^{j\omega/S})
    \label{eq_poly_to_mod_3_omega_2}
\end{equation} 
where $K = lcm(R,S)$. Let ${\bf P}(k)$ be the permutation matrix that rotates ${\bf x}_{m}(e^{j\omega/S})$ upward by $k$. Thus if ${\bf x}^{k}_{m}(z) = {\bf P}(k){\bf x}_{m}(z)$ then \\
${\bf x}^{k}_{m}(z) = [x(zW^{k}_{S}),\cdots,x(zW^{S-1}_{S}),x(z),\cdots,x(zW^{k-1}_{S})]$.  \\

The following lemma gives the relationship of the ideal modulation matrices that give mapping (\ref{eq_map_1}) with those that give mapping (\ref{eq_map_2}).

\begin{lemma}
\label{thm_imm_rel}
If $P$ is even then an ideal modulation matrix ${\bf H}^{2}_{I}(\omega_{1},\omega_{2})$ corresponding to  mapping (\ref{eq_map_2}) can be obtained by the ideal modulation matrix ${\bf H}^{1}_{I}(\omega_{1},\omega_{2})$ corresponding to mapping (\ref{eq_map_1}) as follows,
\begin{equation}
{\bf H}^{2}_{I}(\omega_{1},\omega_{2}) = 
        {\bf H}^{1}_{I}(\omega_{1},\omega_{2}) {\bf P}(R/2) 
    \label{eq_map_1_map_2}
\end{equation}    
where $R \times S$ is the dimension of ideal modulation matrices.
     \label{thm_map_1_map_2}
\end{lemma}

\begin{proof}
Consider the modulation domain equation
\begin{equation}
{\bf y}_{m}(e^{j\omega/R}) 
       = {\bf H}^{1}_{I}(\omega_{1},\omega_{2}){\bf x}_{m}(e^{j\omega/S})
    \label{eq_poly_to_mod_3_omega_ideal}
\end{equation} 
changing $\omega/R$ to $\pi + \omega/R$, we obtain
\begin{equation}
{\bf y}_{m}(e^{j(\pi + \omega/R)}) 
 = {\bf H}^{1}_{I}(\omega_{1},\omega_{2}){\bf P}(R/2){\bf x}_{m}(e^{j\omega/S})
    \label{eq_poly_to_mod_3_omega_ideal_2}
\end{equation} 
Let ${\bf H}^{2}_{I}(\omega_{1},\omega_{2}) = {\bf H}^{1}_{I}(\omega_{1},\omega_{2}){\bf P}(R/2)$. It can be seen that if ${\bf H}^{1}_{I}(\omega_{1},\omega_{2})$ gives mapping corresponding to (\ref{eq_map_1}), then ${\bf H}^{2}_{I}(\omega_{1},\omega_{2})$ gives mapping corresponding to (\ref{eq_map_2}).
\end{proof}

In the following section, we assume without loss of generality that either $P_{1}$ or $P_{2}$ is even, since if both of them are odd then we can double the dimensions of the shift-invariant system as shown in the above lemma.

\section{Filters Characterizing the Modulation Matrix}
\label{sec_modulation_mat}

 
We recall from (\ref{eq_h_mod_2}) that the modulation matrix is given by

\begin{equation}
{\bf H}_{m}(z^{1/K}) 
              = \frac{1}{Q} {\bf W}_{P}^{\dagger}{\bf \Gamma}_{P}(z^{-1/P})
                             {\bf H}_{p}(z)
                             {\bf \Gamma}_{mq}(z^{1/Q}){\bf W}_{Q}   
    \label{eq_h_mod_3} 
\end{equation} 

where $K = lcm(P,Q)$. 

Let $m = gcd(P,Q)$, $p = P/m$ and $q = Q/m$. Then $K = mpq$. Let $(a,b)$ be a pair of integers such that $ap+bq = 1$. An element of ${\bf H}_{m}(z^{1/K})$ is given by 
\begin{equation}
[{\bf H}_{m}(z^{1/mpq})]_{lk} 
 = \frac{1}{mq}\sum_{r=0}^{mp-1}\sum_{s=0}^{mq-1}
               {z^{(ps-qr)/mpq}W_{mpq}^{kps-lqr}H_{rs}(z)}  
    \label{eq_elem_h_mod} 
\end{equation} 

We now have the following lemma.

\begin{lemma}
\label{thm_mod_filt}
The modulation matrix can be characterized by $m$ filters. In perticular we assert the following
\begin{enumerate}
    \item For a given integer $d \in \{0..m-1\}$ there are $mpq$
          pairs $(l,k)$ such that $(l-k)_{m} = d$, where $l \in \{0..mp-1\}$ and 
          $k \in \{0..mq-1\}$.
    \item If $(l,k)$ and $(l',k')$ are two pairs of 
          indices such that $(l-k)_{m} = (l'-k')_{m}$ then  
          \begin{equation}
                [{\bf H}_{m}(z^{1/mpq})]_{l'k'} 
                  = [{\bf H}_{m}(z^{1/mpq}W_{mpq}^{g})]_{lk} 
                \label{eq_h_elem_2} 
         \end{equation} 
         for some integer $g$ such that $0 \le g \le mpq-1$.
    \item $g = 0$ iff $(l',k') = (l,k)$.
\end{enumerate} 
\end{lemma}

\begin{proof}
In order to prove the first part, for a given $l$ we count all the $k$ that give $(l-k)_{m} = d$. Such $k$ are given by $k = ((l-d) + rm)_{mq}, 0 \le r \le q-1$. Thus for a given $l$ we have $q$ values of $k$ that give $(l-k)_{m} = d$. Since $l$ takes values from $0$ to $mp-1$, total number of pairs $(l,k)$ satisfying $(l-k)_{m}$ will be $mpq$.

In order to prove the second part, we note that 
\begin{eqnarray}
[{\bf H}_{m}(z^{1/mpq})]_{l'k'} 
 &=& \frac{1}{mq}\sum_{r=0}^{mp-1}\sum_{s=0}^{mq-1}
               {z^{\frac{(ps-qr)}{mpq}}W_{mpq}^{(k'-k)ps-(l'-l)qr}
                W_{mpq}^{kps-lqr}H_{rs}(z)}  
    \label{eq_hm_elem2} 
\end{eqnarray} 

Let $(l'-l)_{mp} = t$ and $(k'-k)_{mq} = t + hm$, for some integers $t$ and $h$, then it can be shown that
\begin{eqnarray}
W_{mpq}^{(k'-k)ps-(l'-l)qr} 
 &=& W_{mpq}^{g(ps-qr)} 
    \label{eq_w_hm} 
\end{eqnarray} 
where $g = (t+hamp)$. Using this in (\ref{eq_hm_elem2}) we get

\begin{eqnarray}
[{\bf H}_{m}(z^{1/mpq})]_{l'k'} 
 &=& [{\bf H}_{m}(z^{1/mpq}W_{mpq}^{g})]_{lk} 
    \label{eq_hm_elem3} 
\end{eqnarray} 
This proves the second part.

To prove the third part, we note that if $(l',k') = (l,k)$ then $t=0, h=0$ and we have $g = 0$. Now if $g = 0$, then $t + ahmp = fmpq$ for some integer $f$. This gives $t = (fq-ah)mp$. Thus $t$ is a multiple of $mp$. But since $0 \le t \le mp-1$, the only possibllity is $t = 0$, and this gives $l' = l$. In order to prove $k' = k$, we note that with $t = 0$ we have $ahmp = fmpq$. Now since $a$ is not a multiple of $q$, $h$ should be a multiple of q. Let $h = vq$, then $(k' - k)_mq = vmq$, giving $v = 0$, and $k' = k$.  

To show that the modulation matrix can be characterized by $m$ filters, for a $d \in \{0..m-1\}$ choose any index $(l,k)$ such that $(l-k)_{m} = d$. Define $H_{d}(z^{1/mpq}) = [{\bf H}_{m}(z^{1/mpq})]_{lk}$. Then from the above results, $mpq$ elements of the modulation matrix, namely $[{\bf H}_{m}(z)]_{l'k'}$ such that $(l'-k')_{m} = d$ can be written as
\begin{equation}
[{\bf H}_{m}(z^{1/mpq})]_{l'k'} = H_{d}(z^{1/mpq}W_{mpq}^{g})
    \label{eq_hm_hd2} 
\end{equation} 
In this manner all the elements of ${\bf H}_{m}(z^{1/mpq})$ can be written in the form of $H_{d}(z^{1/mpq}W_{mpq}^{g})$ for some $d \in \{0..m-1\}$ and some $g \in \{0..mpq-1\}$. 
\end{proof} 

By using the lemma(\ref{thm_imm}) and the lemma(\ref{thm_mod_filt}) we can obtain ideal frequency spectrum of the $m$ filters that characterize the modulation matrix. The following lemma gives ideal frequency spectrum of these $m$ filters:

\begin{lemma}
\label{thm_idl_mod_filt}
Only maximum two filters out of $m$ filters will have non-zero ideal frequency spectrum. Moreover by appropriate choice of the filters, it can be shown that the filters having non-zero ideal frequency spectrum will have frequency spectrum given by $\chi(\frac{v \pi}{mpq}, \frac{(v + P) \pi}{mpq})$ and $\chi(\frac{-(v+P) \pi}{mpq}, -\frac{v \pi}{mpq})$ for some integer $v$.  
\end{lemma} 

\begin{proof}
In order to prove that only two filters out of $m$ filters have non-zero frequency spectrum, we note that when $P_{1}$ is even, then it can be seen that if ${\bf H}_{I}(-\pi,0)]_{lk} = 1$, then either $(l-k)_{m} = (P_{1}/2)_{m}$ or $(l-k)_{m} = (-P_{1}/2)_{m}$. Let $H_{d1}(z)$ and $H_{d2}(z)$ be two filters out of $m$ filters such that $[{\bf H}_{m}(z^{1/mpq})]_{l,k} = H_{d1}(z^{1/mpq}W_{mpq}^{g})$ when $(l-k)_{m} = (-P_{1}/2)_{m}$ and $[{\bf H}_{m}(z^{1/mpq})]_{l,k} = H_{d2}(z^{1/mpq}W_{mpq}^{g})$ when $(l-k)_{m} = (P_{1}/2)_{m}$, then it can be seen that the ideal frequency spectrum of only $H_{d1}(z)$ and $H_{d2}(z)$ will be non-zero. 

Similarly if $P_{2}$ is even, then $H_{d1}(z)$ and $H_{d2}(z)$ will have nonzero ideal frequency spectrum if $[{\bf H}_{m}(z^{1/mpq})]_{l,k} = H_{d1}(z^{1/mpq}W_{mpq}^{g})$ when $(l-k)_{m} = (-P_{2}/2)_{m}$ and $[{\bf H}_{m}(z^{1/mpq})]_{l,k} = H_{d2}(z^{1/mpq}W_{mpq}^{g})$ when $(l-k)_{m} = (P_{2}/2)_{m}$. Rest of the filters will have zero ideal frequency spectrum.

In order to prove the second part, we consider two cases
\begin{enumerate}
    \item $P_{1}$ is even:
	  For this case if we choose $H_{d1}(z)$ and $H_{d2}(z)$  such that 
          \begin{equation}
	      H_{d1}(z^{1/mpq}W_{mpq}^{v}) 
                    = [{\bf H}_{m}(z^{1/mpq})]_{0,P_{1}/2}
              \label{eq_h1_1} 
          \end{equation} 
          \begin{equation}
              H_{d2}(z^{1/mpq}W_{mpq}^{-v}) 
                    = [{\bf H}_{m}(z^{1/mpq})]_{0, (mq - \frac{P_{1}}{2})}
               \label{eq_h2_1} 
          \end{equation} 

	  then it can be shown that  

          \begin{eqnarray}
          S_{hd1}(-\pi,\pi) &=& \chi(\frac{-(2v+P)\pi}{mpq}, 
                                    \frac{-2v\pi}{mpq}) \nonumber \\
          S_{hd2}(-\pi,\pi) &=& \chi(\frac{2v\pi}{mpq}, 
                                    \frac{(2v+P)\pi}{mpq}) 
              \label{eq_ideal_h1} 
          \end{eqnarray} 

	  If $(P_{1})_{m} = 0$ then we have $d1 = d2$ and for such a case
          \begin{eqnarray}
          S_{hd1}(-\pi,\pi) &=& \chi(\frac{-(2v+P)\pi}{mpq}, 
                                    \frac{-2v\pi}{mpq}) \cup 
                                \chi(\frac{2v\pi}{mpq}, 
                                    \frac{(2v+P)\pi}{mpq}) 
              \label{eq_ideal_h1_2} 
          \end{eqnarray} 
	  with $v = \frac{apP_{1}}{2}$.

    \item $P_{2}$ is even:
	  For this case if we choose $H_{d1}(z)$ and $H_{d2}(z)$  such that 
          \begin{equation}
	      H_{d1}(z^{1/mpq}W_{mpq}^{v}) 
                    = [{\bf H}_{m}(z^{1/mpq})]_{0,P_{2}/2}
              \label{eq_h1_2} 
          \end{equation} 
          \begin{equation}
              H_{d2}(z^{1/mpq}W_{mpq}^{-v}) 
                    = [{\bf H}_{m}(z^{1/mpq})]_{0, (mq - \frac{P_{2}}{2})}
               \label{eq_h2_2} 
          \end{equation} 

	  then it can be shown that  

          \begin{eqnarray}
          S_{hd1}(-\pi,\pi) &=& \chi(\frac{-2v\pi}{mpq}, 
                                    \frac{-(2v-P)\pi}{mpq}) \nonumber \\
          S_{hd2}(-\pi,\pi) &=& \chi(\frac{(2v-P)\pi}{mpq}, 
                                    \frac{2v\pi}{mpq}) 
              \label{eq_ideal_h2_2} 
          \end{eqnarray} 

	  Again if $(P_{2})_{m} = 0$ then we have $d1 = d2$ and for such a case
          \begin{eqnarray}
          S_{hd1}(-\pi,\pi) &=& \chi(\frac{-2v\pi}{mpq}, 
                                    \frac{-(2v-P)\pi}{mpq}) \cup
                                \chi(\frac{(2v-P)\pi}{mpq}, 
                                    \frac{2v\pi}{mpq}) 
              \label{eq_ideal_h2_2_2} 
          \end{eqnarray} 
	  with $v = \frac{apP_{2}}{2}$.
\end{enumerate}    

\end{proof} 

The ideal frequency spectra of these filters are  used to form an objective function, minimizing which would give us desired paraunitary polyphase matrix.

\section{Design of the Paraunitary Polyphase Matrix}
\label{sec_design}

\subsection{Formation of an objective function} 
Given rational numbers $\{\frac{p_{0}}{q_{0}},\cdots, \frac{p_{N-1}}{q_{N-1}}\}$, our objective is to design a paraunitary polyphase matrix ${\bf H}_{p}(z)$ that gives the required splitting of the input frequency spectrum $S_{x}(0,\pi)$into $N$ bands. 

Consider $n$-th channel of such a filter bank. The part of the input spectrum corresponding to this band is $S_{x}(\frac{\bar{p}_{n-1}}{\bar{q}_{n-1}}\pi,\frac{\bar{p}_{n}}{\bar{q}_{n}}\pi)$, where $\frac{\bar{p}_{n-1}}{\bar{q}_{n-1}} = \sum_{i=0}^{n}{\frac{p_{i}}{q_{i}}}$ and $\frac{\bar{p}_{n}}{\bar{q}_{n}} = \sum_{i=0}^{n+1}{\frac{p_{i}}{q_{i}}}$. Let $Q_{n} = lcm(\bar{q}_{n-1},\bar{q}_{n})$ and $P_{n} = Q_{n}\left(\frac{\bar{p}_{n}}{\bar{q}_{n}} - \frac{\bar{p}_{n-1}}{\bar{q}_{n-1}}\right)$. It is clear from lemma \ref{thm_mod_mat_dim} that the ideal modulation matrices for this channel will be of the size $P_{n} \times Q_{n}$ or $2P_{n} \times 2Q_{n}$. Let the dimensions of ideal modulation matrix be denoted by $\bar{P}_{n} \times \bar{Q}_{n}$. 

Since we have to realize a paraunitary polyphase matrix for the filter bank, it is required that the polyphase matrices for all the channels should have same number of columns. Let $S = lcm(\bar{Q}_{0},\bar{Q}_{1},\cdots,\bar{Q}_{N-1})$,$R_{n} = p_{n}S/q_{n}$,$K_{n} = \sum_{j=0}^{i+1}{R_{n}} - \sum_{j=0}^{i}{R_{n}}$ and $M_{n} = gcd(K_{n},S)$. Then the polyphase matrix for $n$-th channel, denoted by ${\bf H}_{pn}(z)$, will be of the dimensions $K_{n} \times S$ and the polyphase matrix for the filter bank will be given by
\begin{equation}
{\bf H}_{p}(z) = \left[
                    \begin{array}{c}
                    {\bf H}_{p0}(z) \\
                    {\bf H}_{p1}(z) \\
                    \vdots          \\ 
                    {\bf H}_{p(N-1)}(z) 
                    \end{array}
                 \right]  
    \label{eq_filtbank_polyphase}
\end{equation} 
 
Let the modulation matrix corresponding to ${\bf H}_{pn}(z)$ be denoted by ${\bf H}_{mn}(z)$. The modulation matrix ${\bf H}_{mn}(z)$ will be characterized by $M_{n}$ filters. Let these filters be denoted by $H_{in}(z), 0 \le i \le M_{n}-1$. Since ${\bf H}_{pn}(z)$ is paraunitary, from (\ref{eq_h_mod_2}) it can be seen that ${\bf H}_{mn}(z)$ is also paraunitary. By using the paraunitary condition on the modulation matrix ${\bf H}_{m}(z)$, i.e. ${\bf H}_{m}(z)\tilde{{\bf H}}_{m}(z) = {\bf I}$, it can be shown that the filters characterizing the modulation matrix satisfy 
\begin{equation}
\sum_{i=0}^{M_{n}-1}{\left|H_{in}(e^{j\omega})\right|^{2}} = K_{n} 
    \label{eq_abs_mod_filt}
\end{equation} 
where $K_{n}$ is some constant. From this, it can be seen that if the stopband energy of all the filters is minimized, then the passband energy would get maximized. Thus we can define an objective function, based on stopband energies of filters in all the channels. This is done as follows:

\begin{enumerate}
     \item Find out the dimensions of the polyphase matrix ${\bf H}_{p}(z)$
           as explained above. 
     \item Define a paraunitary polyphase matrix ${\bf H}_{p}(z,\Theta)$ that
           depends on the parameters $\Theta$. 
     \item Find out the ideal frequency response for all the filters
           $H_{in}(z), 0 \le i \le M_{n}-1, 0 \le n \le N-1$.
     \item Choose a transition bandwidth $\epsilon$. Using this transition
           bandwidth and ideal frequency response of $H_{in}(z)$, find out
           its stopband $\gamma_{in}$
     \item Define the objective function
           \begin{equation}
               D(\Theta) = \sum_{n=0}^{N-1}{\sum_{i=0}^{M_{i}-1}
                                 { \int_{\gamma_{in}}
                                 {\left|H_{in}(e^{j\omega})
                                  \right|^{2} d\omega}}}
                       \label{eq_opt}
           \end{equation} 
\end{enumerate} 

By minimizing $D(\Theta)$ with respect to $\Theta$ we can obtain the required polyphase matrix ${\bf H}_{p}(z,\Theta)$.

\subsection{ Design Examples}

For the design examples given below, we will use Given's factorization for the paraunitary matrices \cite{bk_vetterli_filterbank}. Thus if a paraunitary matrix has size $N \times N$ and order $K$, then the parameter vector $\Theta$ will have $K(N-1) + NC2$ parameters.
 
{\it Example 1}: In this example we consider the rational splitting $\{\frac{2}{5},\frac{1}{5}, \frac{2}{5}\}$.

Since for the third channel we have $P_{1} = 3$ and $P_{2} = 5$, accordng to lemma(\ref{thm_imm}) the size of polyphase matrix for this channel will be $4 \times 10$. Thus the complete filter bank can be realized by a $10 \times 10$ polyphase matrix, with first $4$ rows for first channel, next $2$ rows for second channel and last $4$ rows for the last channel. 

By applying lemma(\ref{thm_mod_filt}) and lemma(\ref{thm_idl_mod_filt}) it can be seen that each channel can be characterized by two filters and only one filter out of these two filter has non-zero ideal frequency spectrum. 

A polyphase matrix is realized using the procedure outlined above with transition bandwidth equal to $pi/20$ and the number of stages in the polyphase matrix equal to $7$. The ideal frequency spectrum and actual frequency spectrum of the filters are given in Fig.(\ref{fig_filts_ex_1}). In these figures stopband of ideal filters is shown at $-60dB$ line.  




\begin{figure*}[t]
\centerline{\includegraphics[scale=0.25]{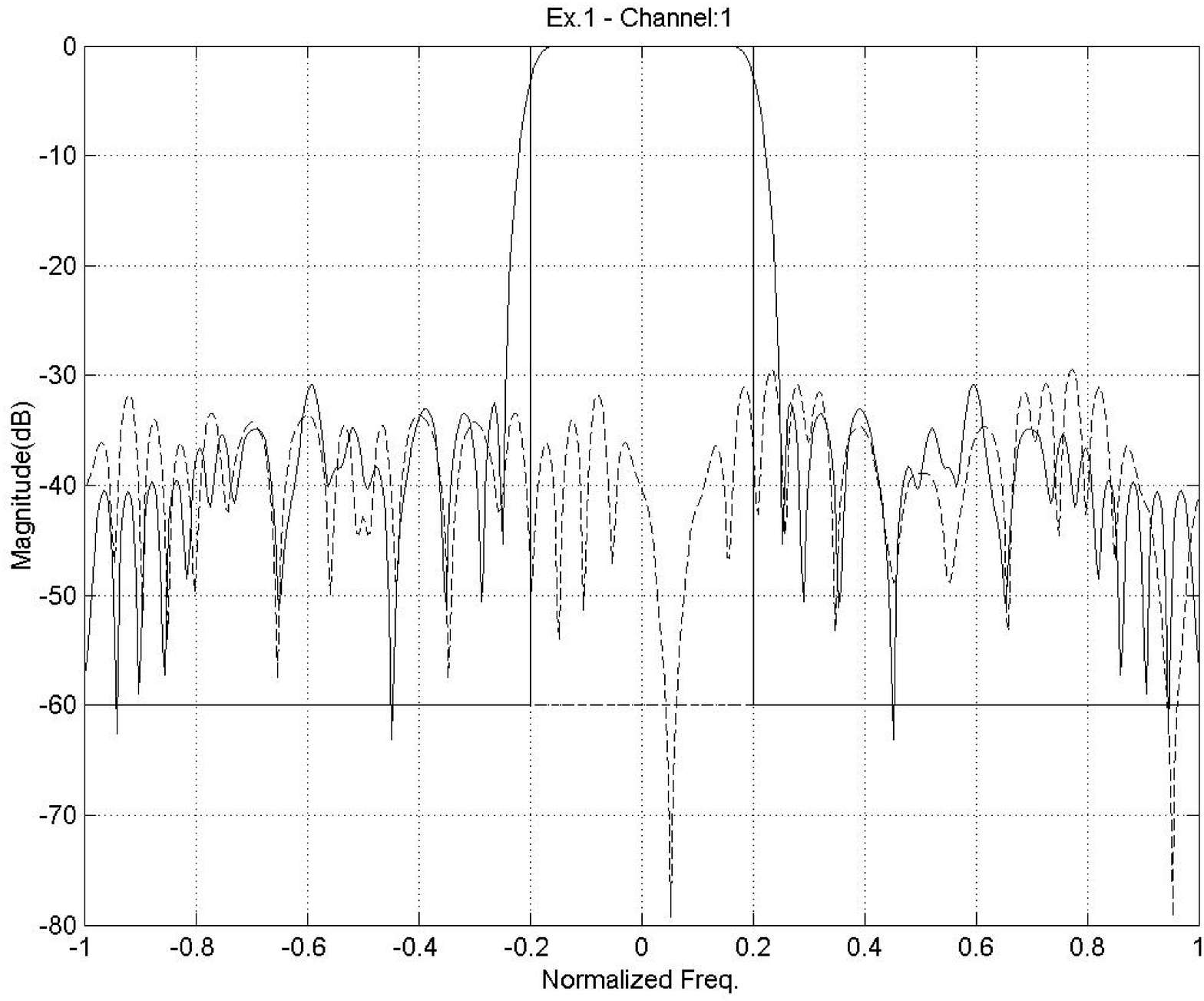}}
\centerline{\includegraphics[scale=0.25]{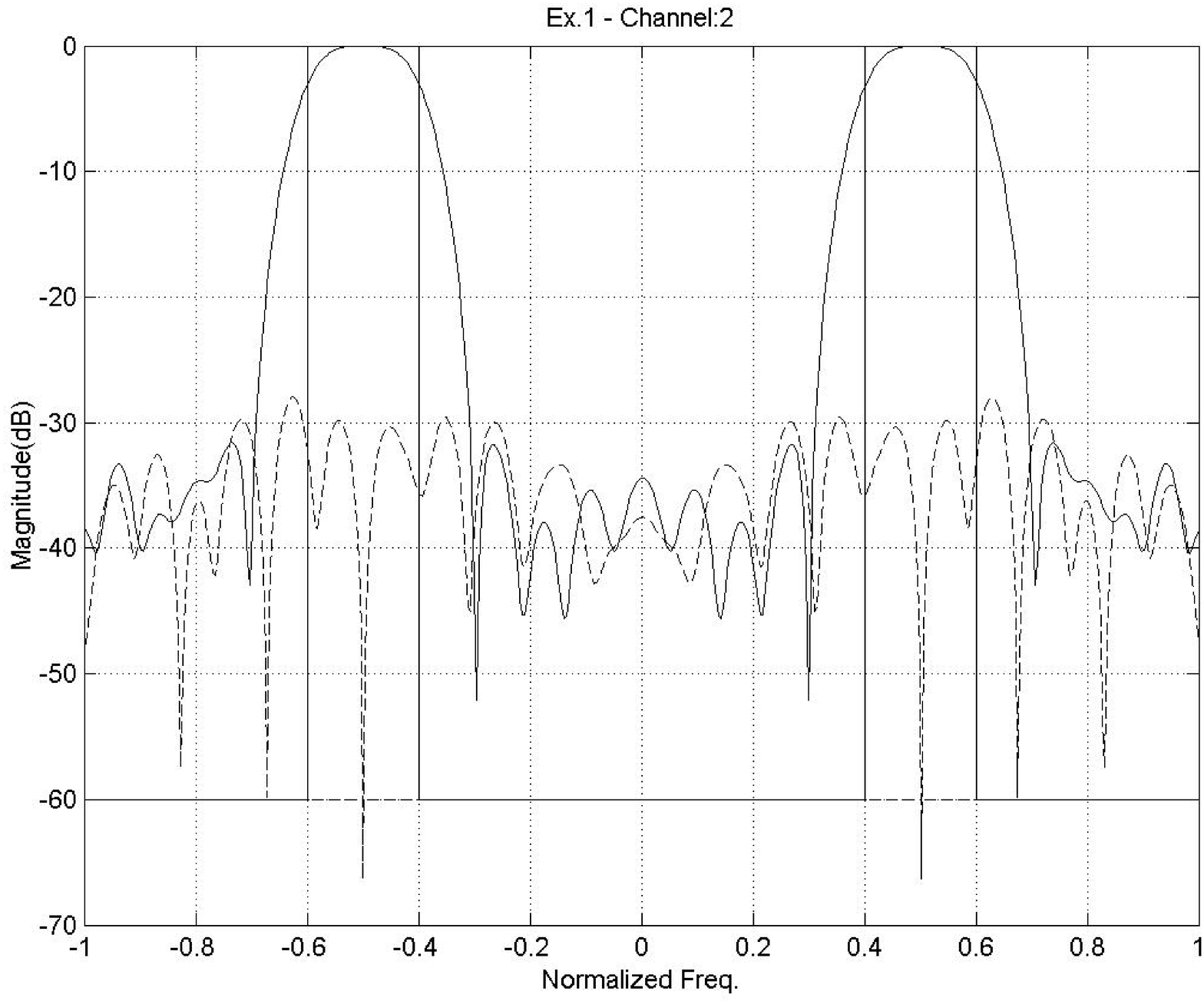},
            \includegraphics[scale=0.25]{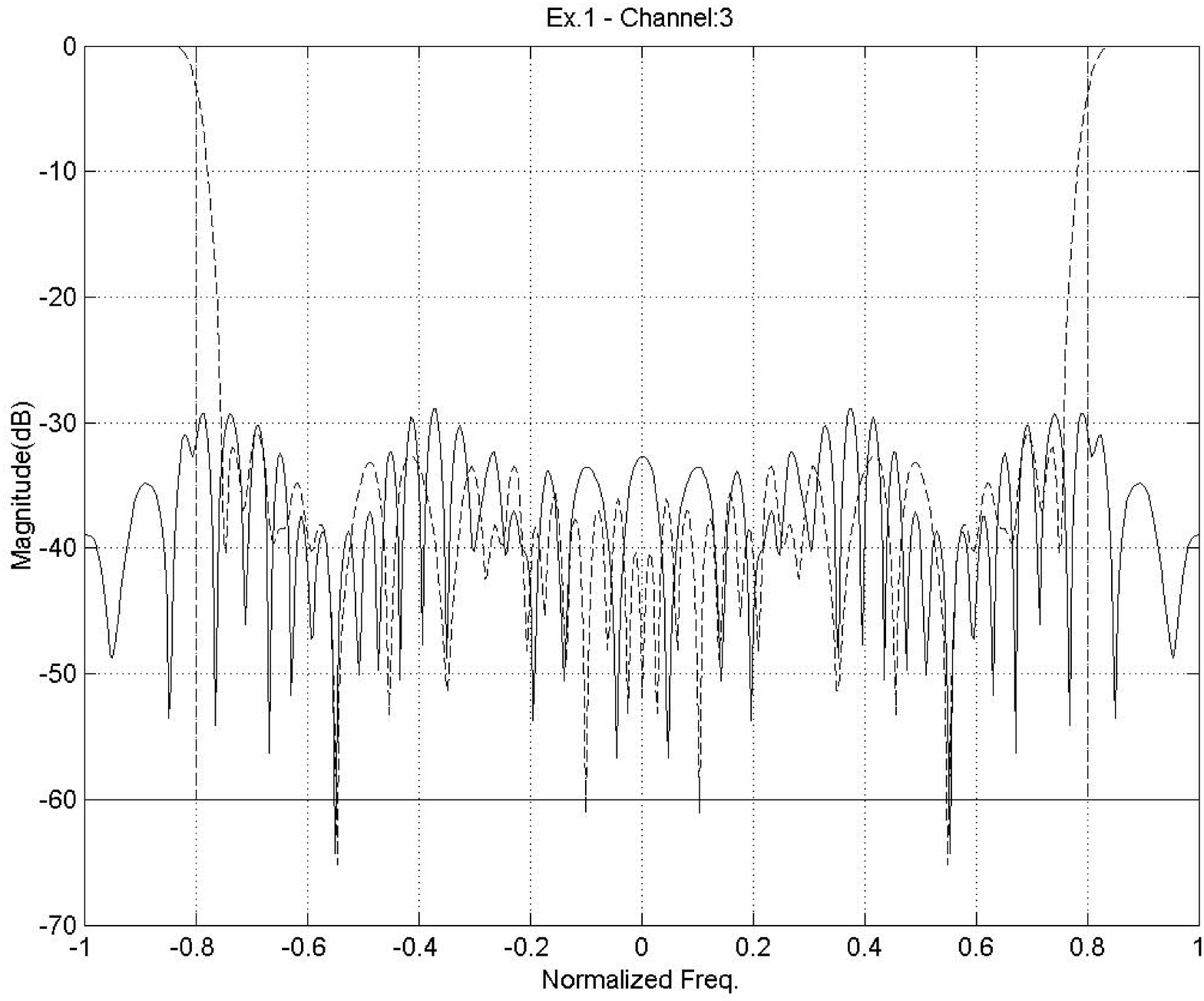}}
\caption{Filters for three channels of the filterbank in Exanple 1 - Ideal frequency spectra are also shown with actual frequency spactra. The stopband of ideal frequency spectra is shown at -60dB}
    \label{fig_filts_ex_1}
\end{figure*}



{\it Example 2}: In this example we consider the rational splitting $\{\frac{2}{9},\frac{1}{3},\frac{1}{3},\frac{1}{9}\}$.

The size of the polyphase matrix in this case would be $9 \time 9$. First channel will have only one filter, next two channels will have three filters and last channel will have only one filter. For second and third channel, two filters will have nonzero frequency spectrum. The ideal and realized filters are shown in Fig.(\ref{fig_filts_ex_1}). For this realization the transition bandwidth is taken as $\pi/20$ and the polyphase matrix has $7$ stages.

\begin{figure*}[t]
\centerline{\includegraphics[scale=0.25]{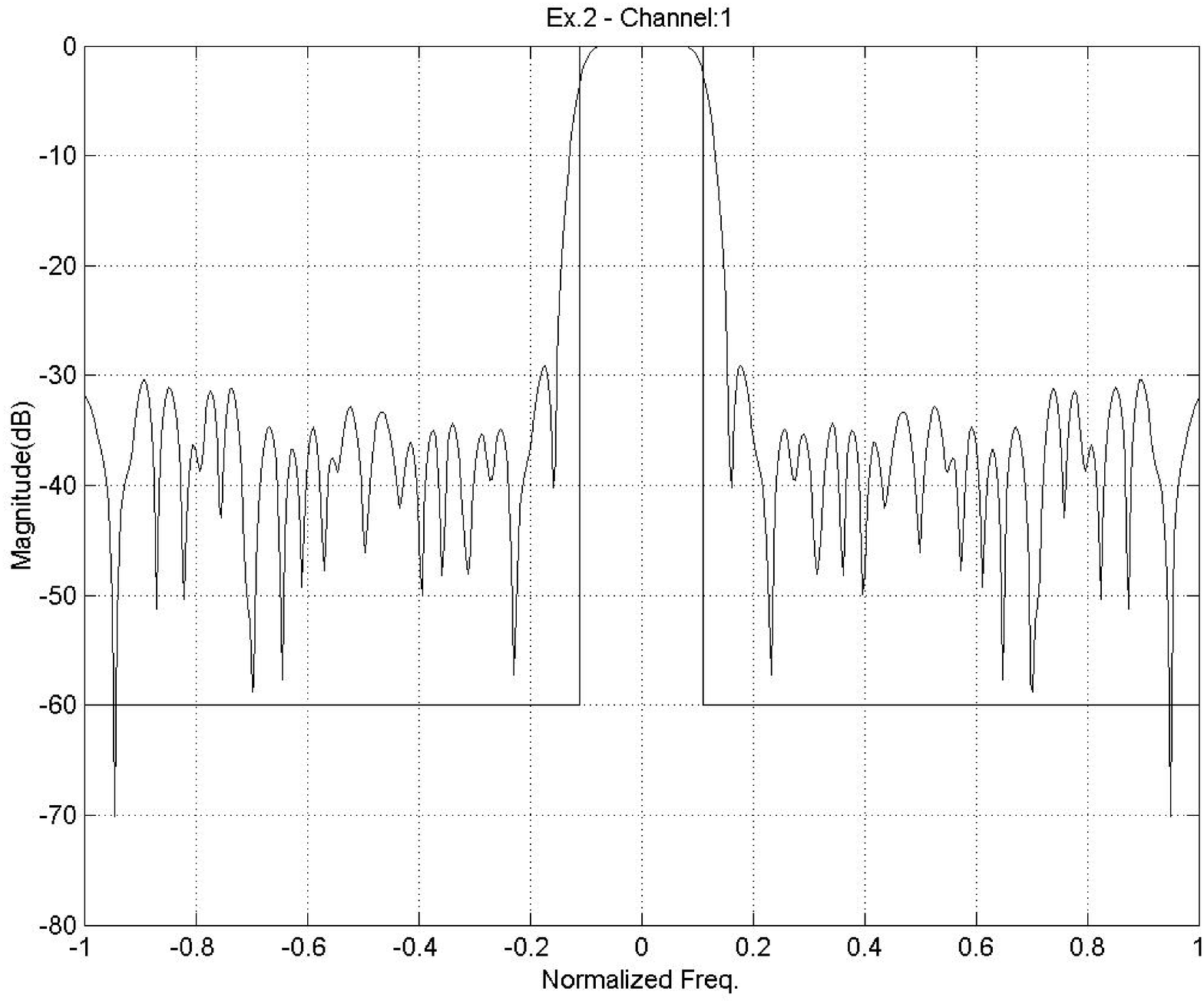},
            \includegraphics[scale=0.25]{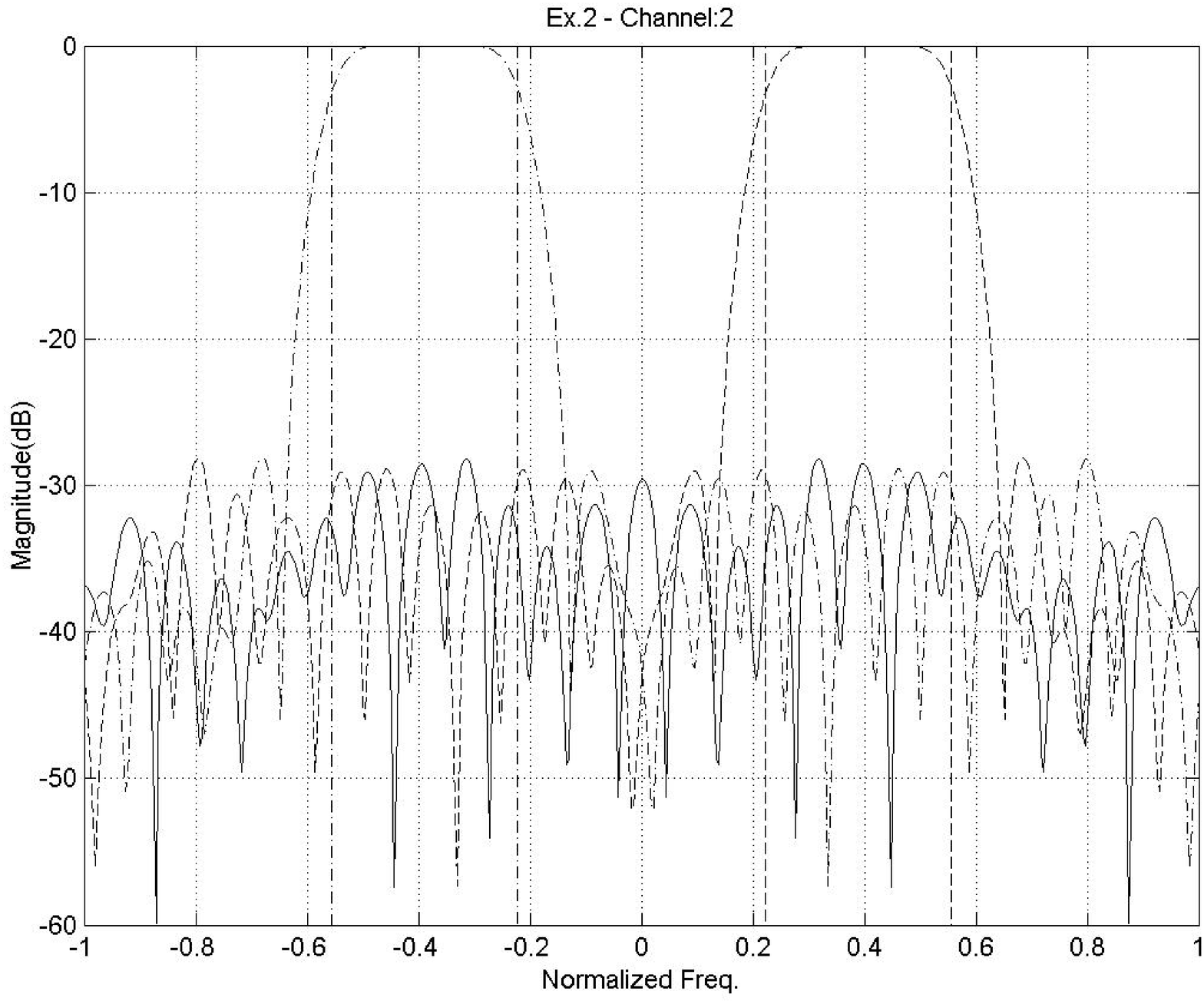}}
\centerline{\includegraphics[scale=0.25]{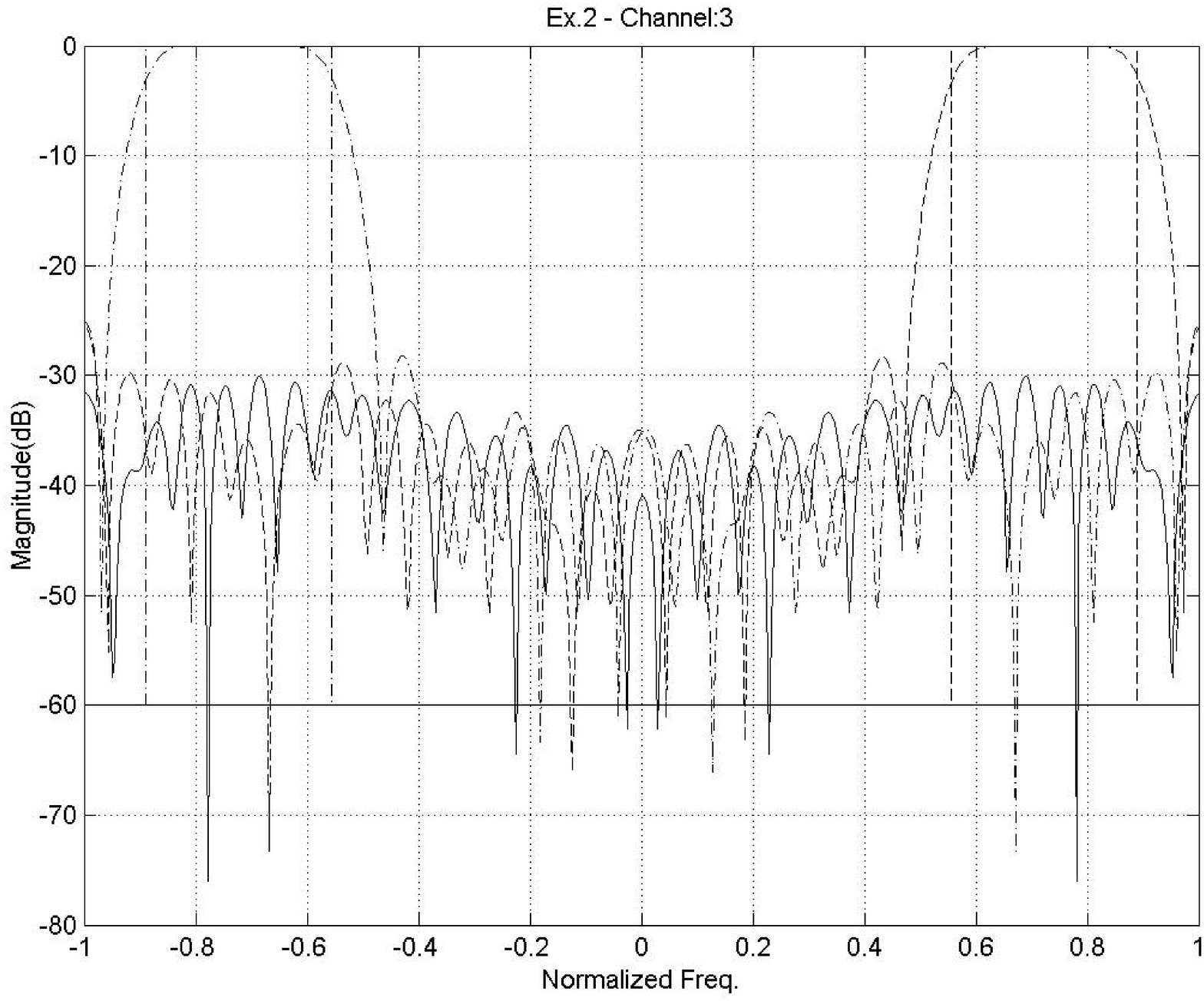},
            \includegraphics[scale=0.25]{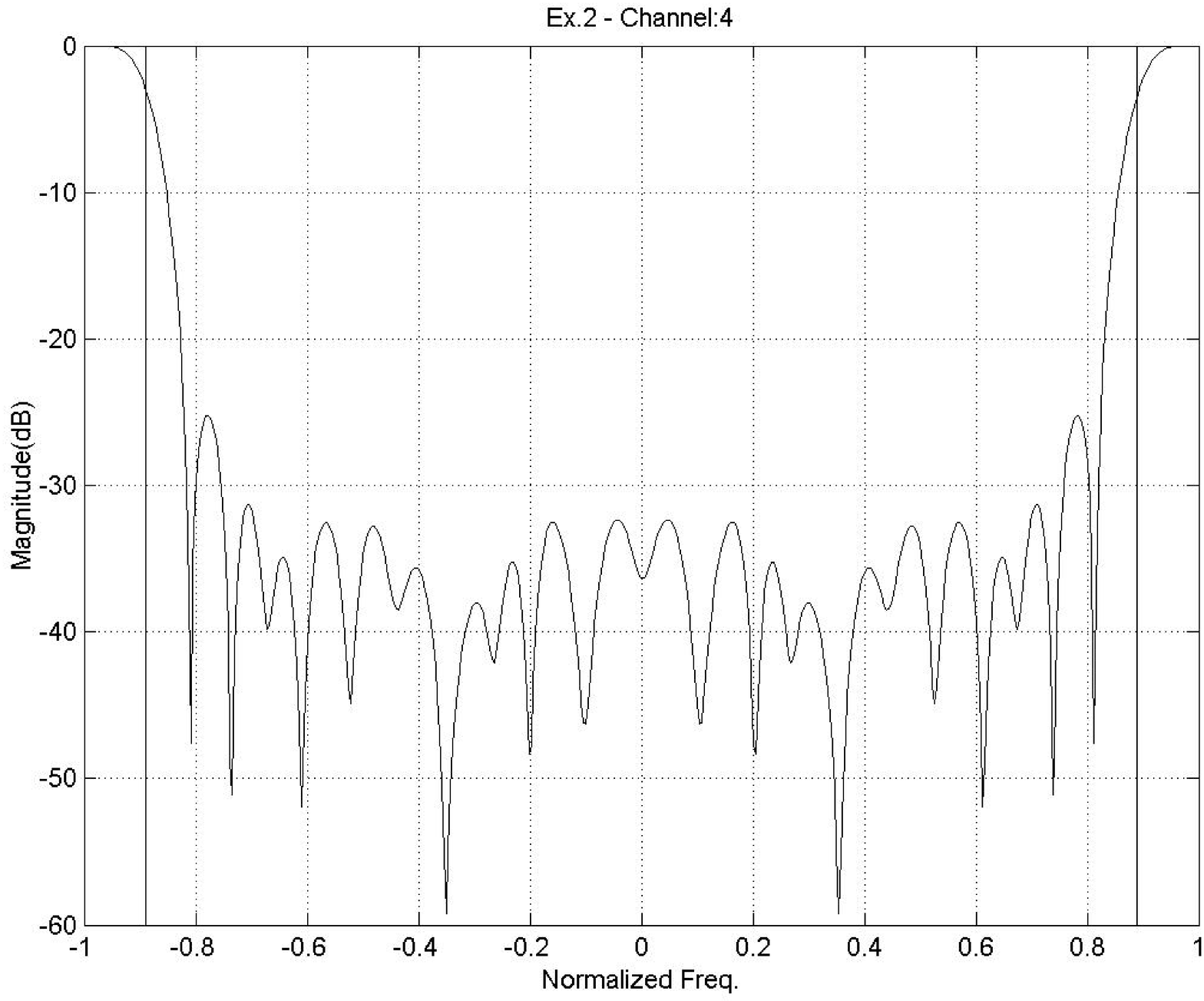}}
\caption{Filters for Example Two - Ideal frequency spectra are also shown with actual frequency spectra. The stopband of ideal frequency spectra is shown at -60dB}
    \label{fig_filts_ex_2}
\end{figure*}

\section{Conclusion}
\label{sec_conclude}

In this paper we have presented a design method to design rational filter banks using $(P,Q)$ shift-invariant systems. We have obtained minimum dimension of a $(P,Q)$-shift invariant system for ideal mapping. Next we have obtained a set of $gcd(P,Q)$ filters characterizing $(P,Q)$ shift-invariant systems. Ideal frequency spectrum of these filters are obtianed obtaining and then paraunitary polyphase matrix of the filter bank is obtained by minimizing the stopband energies of these filters.

Since the design method is based on $(P,Q)$-shift-invariant systems, filter banks with arbitrary rational splitting can be designed by this method. Further improvements in the method are possible by finding out a good ``starting point'' for objective function minimization. 


\bibliographystyle{plain}


\begin{thebibliography}{1}

\bibitem{art_vetterli_rationalFB}
J.~Kovacevic and M.~Vetterli.
\newblock {Perfect Reconstruction Filter Banks with Rational Sampling Factors}.
\newblock {\em IEEE Trans. on Sig. Proc.}, 41(6):2047--2066, 1993.

\bibitem{art_mehr_representations}
A.~S. Mehr and T.~Chan.
\newblock Representations of linearly periodically time-varying and multirate
  systems.
\newblock {\em IEEE Trans. on Signal Processing}, 50(9):2221--2229, 2002.

\bibitem{art_shenoy_multirate}
R.~G. Shenoy.
\newblock {Multirate Specifications via Alias-Component Matrices}.
\newblock {\em IEEE Trans on Ckt. and Sys.-II}, 45(3):314--320, 1998.

\bibitem{art_chen_general}
E.~W.~Bai T.~Chan, L.~Qiu.
\newblock {General Multirate Building Structures with Application to Nonuniform
  Filter Banks}.
\newblock {\em IEEE Trans on Ckt. and Sys.-II}, 45(8):948--958, 1998.

\bibitem{bk_vetterli_filterbank}
M.~Vetterli and J.~Kovacevic.
\newblock {\em {Wavelets and Subband Coding}}.
\newblock Prentice Hall PTR, Englewood Cliffs,New Jersey, 1995.

\end{thebibliography}

\end{document}